\documentclass[11pt]{article}
\usepackage{booktabs}
\usepackage{authblk}
\usepackage{subfig}
\usepackage{amsmath}
\usepackage{amsthm}
\usepackage{amsfonts}
\usepackage{tabularx}
\usepackage{cite}
\usepackage{algorithm}
\usepackage{algorithmic}

\usepackage[pdftex]{graphicx}
\graphicspath{{../Figures/}}
\DeclareGraphicsExtensions{.pdf,.jpeg,.png}

\newtheorem{theorem}{Theorem}
\newtheorem{lemma}{Lemma}
\newtheorem{definition}{Definition}

\advance \oddsidemargin by -0.8in
\newcommand{\myparskip}{3pt}
\parskip \myparskip

\begin{document}
\title{Vulnerability of D2D Communications from Interconnected Social Networks}
\author[1]{Tianyi Pan}
\author[1]{Md Abdul Alim}
\author[1]{Xiang Li}
\author[1]{My T. Thai}

\affil[1]{University of Florida, Gainesville, FL, USA

\{tianyi,alim,xixiang,mythai\}@cise.ufl.edu}

\begin{titlepage}
\maketitle
\thispagestyle{empty}

\begin{abstract}
In this paper, we study how rumors in Online Social Networks (OSNs) may impact the performance of device-to-device (D2D) communication. As D2D is a new technology, people may choose not to use it when believed in rumors of its negative impacts. Thus, the cellular network with underlaying D2D is vulnerable to OSNs as rumors in OSNs may decrement the throughput of the cellular network in popular content delivery scenarios. To analyze the vulnerability, we introduce the problem of finding the most critical nodes in the OSN such that the throughput of a content delivery scenario is minimized when a rumor starts from those nodes. We then propose an efficient solution to the critical nodes detection problem. The severity of such vulnerability is supported by extensive experiments in various simulation settings, from which we observe up to $40\%$ reduction in network throughput.
\end{abstract}

\end{titlepage}
\section{Introduction}\label{sec:intro}
D2D communication is a promising approach to cope with the rapidly increasing demand of mobile data \cite{cisco2016global}, in which user equipments (UEs) directly communicate with each other while bypassing the cellular network's base stations (BSs). For the cellular network, utilizing D2D communication can therefore offload its mobile data traffic and significantly boost the overall performance \cite{doppler2009device,yu2011resource,fodor2012design,lei2012operator}.

However, as a new technology, D2D communication will be likely to face doubts from various perspectives, such as efficiency, safety, etc. When users lacking the knowledge of D2D, exaggerated disadvantages of D2D can also be generated and spread as rumors, either randomly by some normal users, or intentionally by malicious individuals. With OSN as the medium, rumors can propagate conveniently and affect a large portion of the users in the network \cite{nguyen2010novel,nguyen2012containment,fan2013least} and cause the users to opt-out of D2D. When a requester chooses not to use D2D, its request has to be served by a cellular link; when a relay device disables D2D, some requesters in its proximity may not be served via D2D anymore. Both situations can potentially increase the burden of the BS and the overall throughput can degrade. Therefore, the cellular network is vulnerable to the rumors spreading in its interconnected OSN. An example is depicted in Fig.~\ref{fig:social-d2d}.

\begin{figure}[!ht] \centering
\subfloat[D2D link between devices $1,2$ helps offloading traffic from the BS.]{
 	\includegraphics[width=0.45\linewidth]{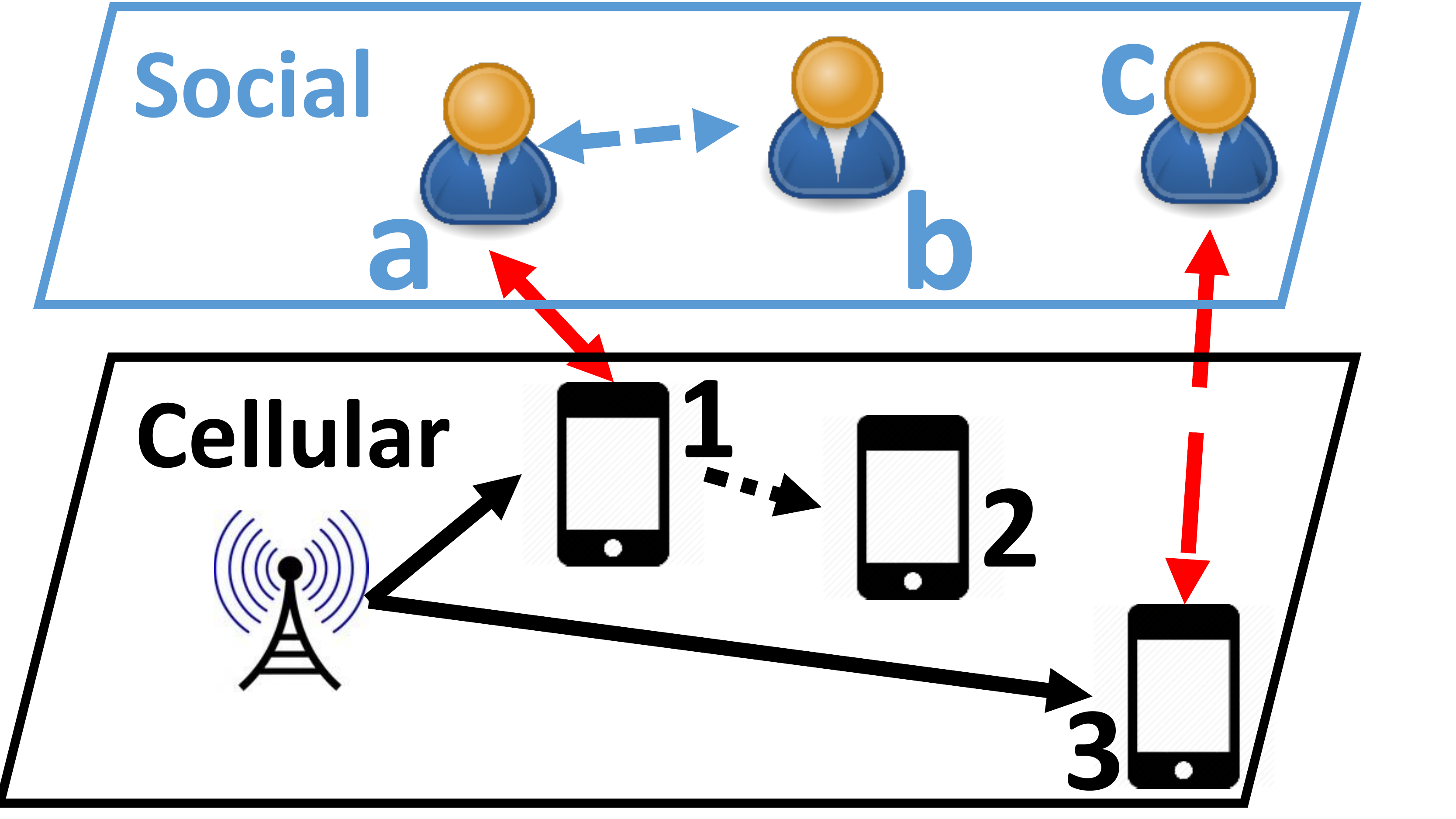}
 	\label{fig:beforemis}}
    \hspace{.09cm}
\subfloat[Throughput drops when device $1$ opt-out of D2D due to rumor.]{ 	\includegraphics[width=0.45\linewidth]{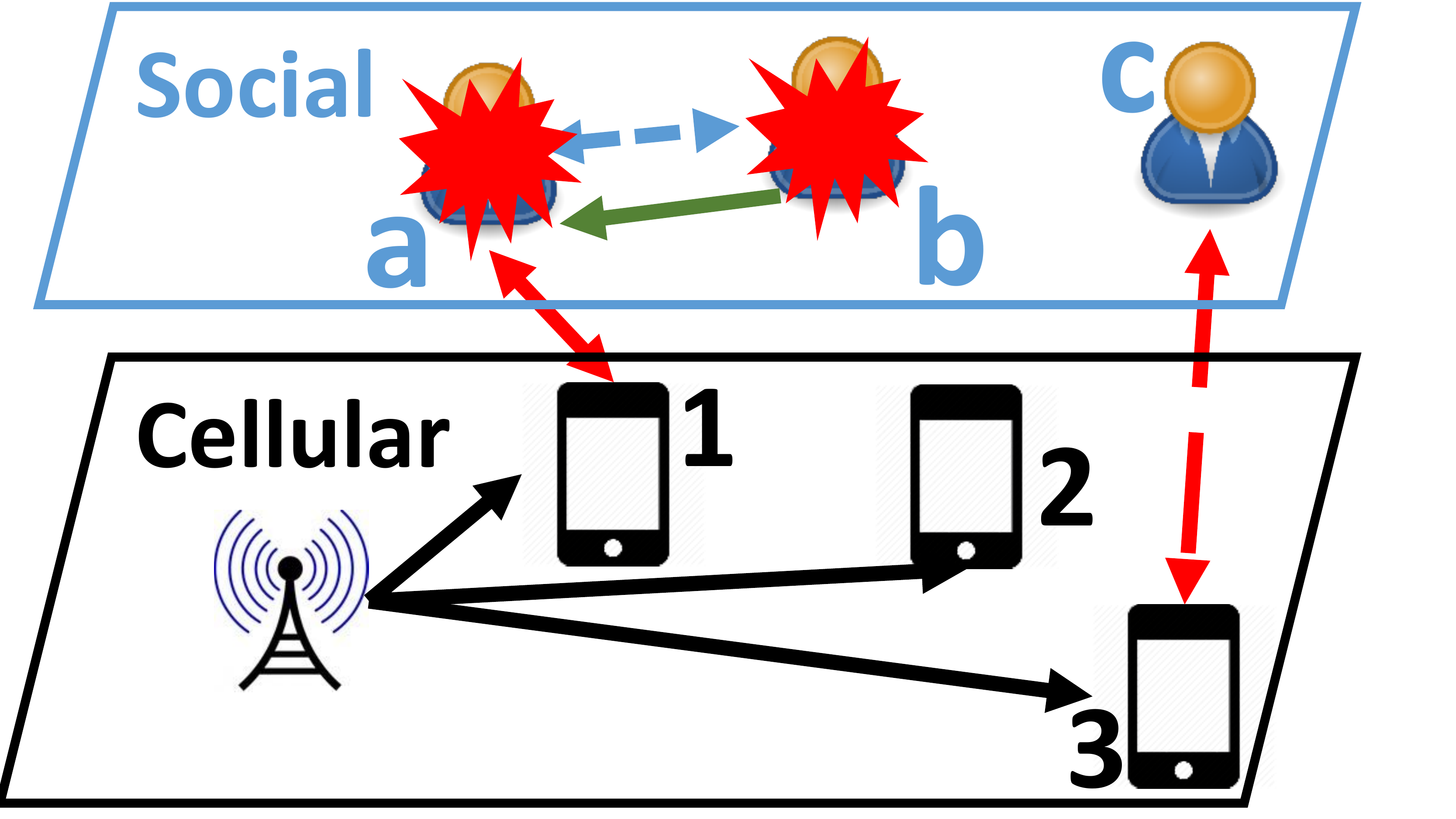}
 	\label{fig:aftermis}}
    \caption{Rumor Impacts D2D Performance}
    \label{fig:social-d2d}
 \end{figure}
 
To assess the vulnerability of the cellular network from an interdependent OSN, we propose the TMIN (Throughput Minimization in INterdependent D2D/OSN) problem that asks for a set of $k$ nodes in the OSN, such that when a rumor starts from those nodes, the throughput in the cellular network for a popular content delivery scenario is minimized. In order to consider the vulnerability in the \textit{worst} case, we assume the rumor starts early and all OSN users believed in the rumor will disable D2D on their devices \textit{before} the content delivery starts. The throughput of the content delivery scenario in the cellular network is then calculated based on the (possibly) reduced number of D2D devices. The main challenge of TMIN is that the data transmission problem in the cellular network and the information propagation problem in the OSN must be jointly considered in order to accurately depict the impact or rumors to D2D communication. With high complexity from both problems, it is difficult to optimally solve TMIN.

To cope with the challenges, we propose Rumor the Critical Framework (RCF) that connects the two networks by first evaluating the criticality of UEs in content delivery and mapping the values to the OSN, and then efficiently solve the Influence Maximization (IM) problem in the OSN, considering node criticality. 

In summary, our contributions are as follows:

\begin{itemize}
\item We analyze the vulnerability of a cellular network with underlaying D2D from an interdependent OSN, by solving the problem TMIN of finding the most critical nodes in the OSN to the throughput in the cellular network.
\item We propose the algorithm RCF that can efficiently make near optimal decisions for TMIN.
\item We experimentally evaluate the vulnerability in various simulation settings with realistic cellular network parameters and real-world OSN data. 
\end{itemize}

\textbf{Related Work.} 
Recently, the use of D2D communication to improve the
spectrum efficiency and system capacity has received much interest \cite{doppler2009device,yu2009performance,janis2009interference,madan2010cell,xu2010effective,xing2010investigation,yu2011resource,min2011capacity,han2012mobile,lei2012operator,fodor2012design,pei2013resource,wang2013energy,wang2013joint,zhang2013distributed,lee2013resource,xu2013efficiency,chen2015exploiting,zhang2015social,li2016optimal}. Towards this direction, previous works studied resource allocation \cite{yu2011resource,wang2013joint,zhang2013distributed,xu2013efficiency,lee2013resource}, interference management \cite{janis2009interference,xu2010effective,madan2010cell,min2011capacity}, power control \cite{wang2013energy,yu2009performance,xing2010investigation} and base station scheduling \cite{li2016optimal}. In order to construct stable D2D transmission links, the social aspect can be utilized. Social trust and social reciprocity have been leveraged in \cite{chen2015exploiting} to enhance cooperative D2D communication based on a coalitional game. In \cite{zhang2015social}, encounter history of the devices as well as their social relations are exploited in order to deploy a reliable D2D communication mechanism. Social communities \cite{nguyen2011overlapping} were also utilized \cite{alim2016leveraging} to facilitate finding stable D2D links. However, none of the existing works can capture how the information diffusion phenomenon in OSNs affects different performance metrics of  D2D communication underlaid cellular networks.

In OSNs, the NP-Complete IM problem \cite{Kempe03} is widely studied. Since the seminal work by Kempe et al. \cite{Kempe03}, a plethora of works for IM have been proposed \cite{Leskovec07,Chen09,Goyal11a,chen2010scalable,Goyal11b,jung2012irie,dinh2012cheap,dinh2014cost,borgs14reverse,Tang14,Tang15,nguyen2016targeted,nguyen2016stop,kuhnle2017scalable,li2017approximate}. The state-of-the-art works \cite{borgs14reverse,Tang14,Tang15,nguyen2016targeted,nguyen2016stop,li2017approximate} follow the RIS approach proposed in \cite{borgs14reverse}. They guarantee the $(1-\frac{1}{e}-\epsilon)$ approximation ratio with a time complexity near linear to the network size. The existing works usually focus on the OSNs, without considering how the information can influence other networks, like D2D networks, that are interconnected with the OSNs.

\textbf{Organization.} 
The rest of the paper is organized as follows. In Sect.~\ref{sec:system}, we describe the cellular network, the OSN, their dynamics, interactions and vulnerability. We also formally define the TMIN problem. We propose our solution RCF to TMIN in Sect.~\ref{sec:solution}. The experiment results are illustrated in Sect.~\ref{sec:experiment}. We conclude the paper in Sect.~\ref{sec:conclusion}.

\section{System Model}\label{sec:system}
In this section, we present the system model of interest and introduce the two components, the cellular network and the OSN respectively. Then we formally define the TMIN problem. 

\subsection{The Cellular Network}\label{ssec:cellularnetwork}
\begin{figure}[!ht] \centering
 	\includegraphics[width=0.8\linewidth]{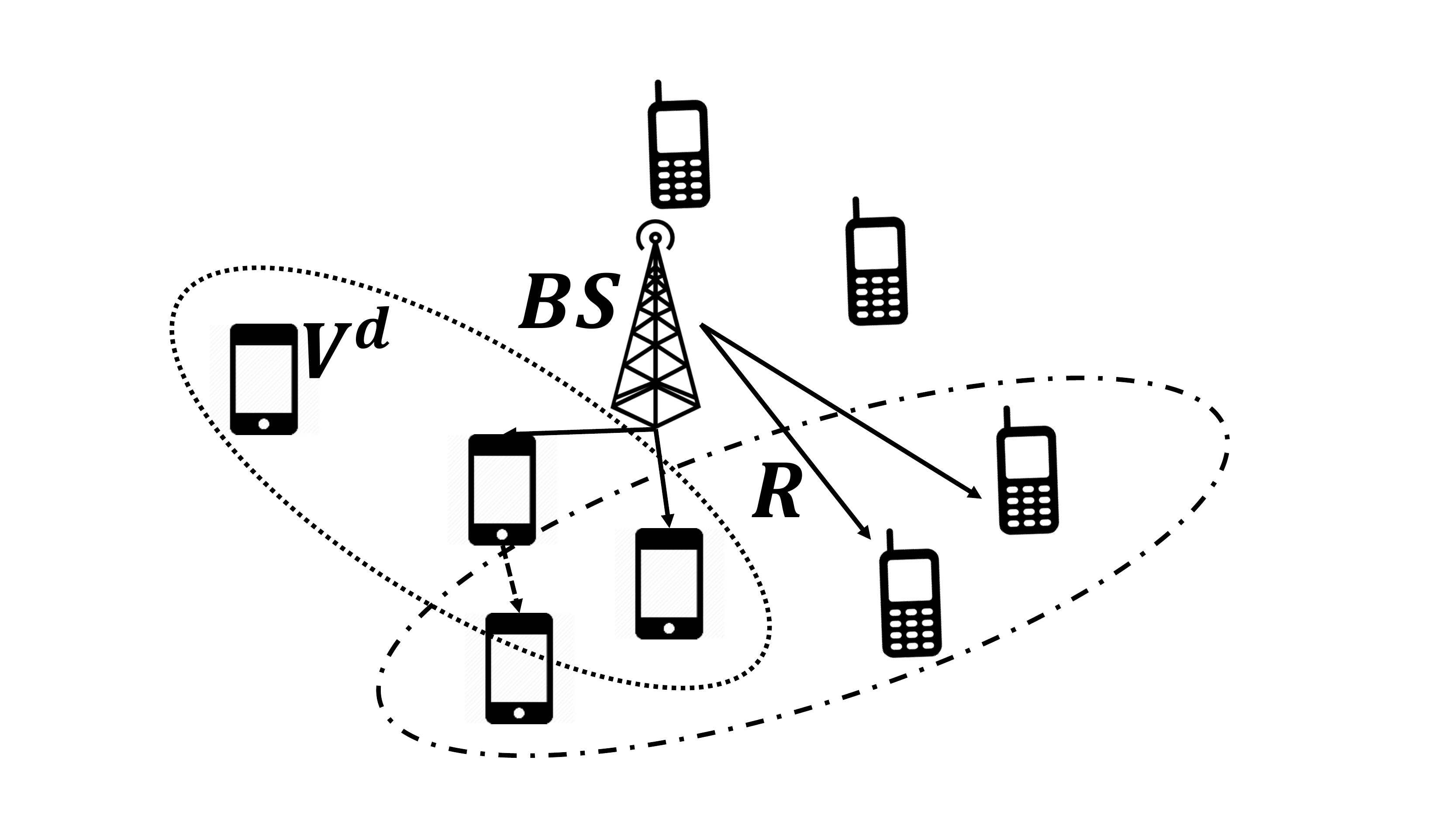}
 	\caption{The Cellular Network}
 	\label{fig:cellularnetwork}
 \end{figure}
We study a cellular system as illustrated in Fig.~\ref{fig:cellularnetwork}. It has a single BS $B$ and a set $\mathcal{D}$ of UEs. We denote the cellular network as $G^c=(V^c, E^c)$, where $V^c = \mathcal{D}\cup\{B\}$ and $E^c$ denotes all cellular/D2D links. A subset of UEs $V^d\subseteq \mathcal{D}$ enables D2D. Denote $G^d=(V^d,E^d)$ as the network induced in $G^c$ by $V^d$. A link $(i,j)\in E^d$ exists if and only if UEs $i,j\in V^d$ are within D2D communication range. A subset $R\subseteq \mathcal{D}$ of UEs request the same content. $R=R^c\cup R^d$ where $R^c = R\backslash V^d$ and $R^d = R\cap V^d$. As a UE in $R^c$ disables D2D, its request can only be served by cellular links. Instead, a UE in $R^d$ can have its request fulfilled either by cellular or D2D links. The set of UEs $V^r = V^d\backslash R^d$ does not request the data, but are able to be relay devices for UEs in $R^d$. For a node $i$, denote $N^{-}(i), N^{+}(i)$ as the set of incoming/outgoing links of $i$, respectively. 

\subsubsection{Cellular Resources}\label{sssec:cellularresources}
In this paper, we consider the resource sharing model discussed in \cite{yu2011resource,lei2012operator,li2016optimal}. In the model, the D2D and cellular links use disjoint portion of the licensed band. Denote the total bandwidth of the BS as $W$. If a resource allocation scheme allocates $W^c$ of the band to cellular links, the bandwidth for D2D links is then $W-W^c$. Therefore, interference only exists among D2D links. Additionally, there is no further division in the frequency domain, so that cellular/D2D links will use the full bandwidth $W^c$/$W^d$. 

\subsubsection{Data Rates}\label{sssec:datarate}
Based on the resource sharing model, a cellular link $(B,i)$ receives no interference from D2D links and uses the full bandwidth. We can express its data rate $r(B,i)$ under distance-dependent path loss and multipath Rayleigh fading as in \eqref{eqn:cellularrate}, where $\gamma_{B,i}$ is the Signal to Noise Ratio (SNR) for $(B,i)$. 

\begin{align}
&r(B,i) = W^c\log_2(1+\gamma(B,i))\label{eqn:cellularrate}\\
&\gamma_{B,i} = \frac{p_B d_{Bi}^{-\alpha}|m_0|^2}{N_0}\label{eqn:snr}
\end{align}

In $\gamma_{B,i}$, $p_B$ is the transmit power of the BS; $d_{Bi}$ is the distance between $B,i$; $\alpha$ is the path loss exponent; $m_0$ is the fading component and $N_0$ is the power of the receiver noise, which is assumed to be the additive white Gaussian noise.

For a D2D link $(j,k)$, we must include interference from other D2D links when calculating data rate $r(j,k)$. Denote $\mathcal{L}$ as the set of D2D links that transmit at the same time with $(j,k)$, we can calculate $r(j,k)$ using \eqref{eqn:d2drate}.

\begin{align}
&r(j,k) = W^d\log_2(1+\gamma(j,k))\label{eqn:d2drate}\\
&\gamma(j,k) = \frac{p_j d_{jk}^{-\alpha}|m_0|^2}{\sum_{(j',k')\in \mathcal{L}}p_{j'}d_{j'k}^{-\alpha}|m_0|^2 +N_0} \label{eqn:sinr}
\end{align}

$\gamma(j,k)$ is the Signal to Interference and Noise Ratio (SINR) and $p_j,p_{j'}$ are the transmit powers for UEs $j,j'$ respectively. 

The model can be extended to handle multiples BSs by adding intercell interference to the denominator part of \eqref{eqn:snr} and \eqref{eqn:sinr}. We model device mobility (when necessary) as multiple snapshots of static cellular networks for popular content delivery scenarios, since such scenarios usually happen at locations where devices have limited mobility (stadium, shopping mall) or the mobility follows certain routine (office, school).

\subsection{The Online Social Network}
We abstract the OSN to be a weighted directed graph $G^s$ with a node set $V^s$ and a directed link set $E^s$, where a node $v\in V^s$ represents a user. A link $(u,v)\in E^s$ exists if and only if node $v$ follows node $u$ in the OSN. Also, each $(u,v)$ is associated with a weight $p(u,v)\in [0,1]$ for information propagation. 
\subsubsection{Information Propagation}
To characterize how rumor propagates in the OSN, we will focus on the Independent Cascading (IC) model \cite{Kempe03} in this paper. However, our results can be easily extended to the Linear Threshold (LT) model. 

In the IC Model, initially no nodes believe the rumor, we term this as the unactivated status. Given a seed set $S$, the rumor propagates in rounds. In round $0$, only the nodes $v\in S$ are activated (believed in the rumor). In round $t\geq 1$, all nodes activated at round $t-1$ will try to activate their neighbors. An activated node $u$ will remain activated. It has probability $p(u,v)$ to activate each unactivated neighbor $v$ at the next round and it cannot activate any neighbors afterwards. The process stops when no more nodes can be activated. Denote $\mathbb{I}(S)$ as the expected number of nodes activated by the rumor with $S$ as the seed set, where the expectation is taken over all $p(u,v)$. We call $\mathbb{I}(S)$ as the influence spread of seed set $S$. When a problem considers differentiated gain of influencing the nodes, we denote the expected gain over all influenced nodes by $S$ as $\mathbb{I}^c(S)$.

\subsubsection{Interconnection between OSN and the Cellular Network}
The owners of the UEs $\mathcal{D}$ in the cellular network may be users in an OSN, as illustrated in Figure \ref{fig:social-d2d}. As our focus is on the impact of OSNs to D2D communication, we construct a link $e^{vi}$ only if a UE $i\in V^d$ has a corresponding user $v\in V^s$. The collection of all such edges is denoted as $E^{ds}$. We assume that each UE can be related to at most one OSN user and each OSN user owns at most one UE. Otherwise, dummy nodes can be used to recover the one-to-one correspondence. If $e^{vi}\in E^{ds}$ and $i\in V^d$, when $v$ believed in the rumor, UE $i$ will no longer be in $V^d$ and cannot be included in any D2D links. 

\begin{definition}[TMIN]
Consider a cellular system with BS $B$, set of devices $\mathcal{D}$, set of content requesters $R\subset \mathcal{D}$ and the cellular network $G^c=(V^c,E^c)$. Also, consider the OSN $G^s=(V^s,E^s)$ and the correspondence among devices in $V^d$ and users in $V^s$ depicted by the edge set $E^{ds}$. The information propagation is described using the IC model with probability $p_{uv}$ for each $(u,v)\in E^s$. TMIN asks for a seed set in $G^s$ with size at most $k$ to minimize the throughput $T$ in $G^c$.
\end{definition}

Solving TMIN reveals the top-$k$ nodes in the OSN that are critical to the cellular network. Also, comparing the throughput in a scenario that is not impacted by rumors and the one solved by TMIN helps characterizing the degree of damage that rumors can cast to the cellular network.

\section{Solution to TMIN}\label{sec:solution}
In this section, we describe our solution RCF to TMIN by first presenting its overview and then introduce its two subroutines in detail. 
\subsection{The Overview}
Intuitively, the goal for TMIN is to find the most critical nodes in the OSN as seed nodes for the rumor, in order to impact the cellular network throughput the most. However, the OSN alone contains no information about the cellular network. Therefore, we must utilize the cellular network to provide information to the OSN and guide its critical nodes selection. In our solution framework RCF, we first calculate the criticality of all UEs in the cellular network using Alg.~\ref{alg:nce}, based on the solutions for a throughput minimization problem defined in the cellular network. Then, we cast the criticality values of the devices to their corresponding users in the OSN and introduce an efficient targeted IM algorithm (Alg.~\ref{alg:tim}) to find the top-$k$ critical nodes in the OSN. The framework RCF is described in Alg.~\ref{alg:rcf}, while the two subroutines, Alg.~\ref{alg:nce} and Alg.~\ref{alg:tim} are discussed in Sect.~\ref{ssec:D2Donly} and \ref{ssec:bothnetwork}, respectively.

\begin{algorithm}
 	\caption{\textbf{RCF}}
 	\label{alg:rcf}
 	\begin{algorithmic}
 	    \REQUIRE Cellular Network snapshots $G^{c}=(V^{c},E^{c})$, Social network $G^s=(V^s,E^s)$,$k$
 	    \ENSURE Seed set $S\subseteq V^s$
 	    \STATE Calculate criticality values $\mathcal{C} = \{cr_i, |\forall i\in v^d\}$ using Alg.~\ref{alg:nce} with budgets $k,k+1,\cdots,|V^d|$ and project them to corresponding nodes in $V^s$.\\
 	 	\STATE	Solve the seed set $\bar{S}$ by Alg.~\ref{alg:tim}.\\
 	 	\RETURN $\bar{S}$
 	\end{algorithmic}

 \end{algorithm}

\subsection{Criticality Evaluation Scheme}\label{ssec:D2Donly}
As the goal of the cellular network is to maximize its throughput, the criticality of the devices must be related to their contribution in throughput reduction. To characterize criticality, we start from the problem of finding the top-$u$ critical nodes in the cellular network: those whose switching from D2D to cellular mode can minimize the maximum network throughput. With a fixed $u$, we are able to obtain the top-$u$ critical nodes, yet with this piece of information alone we can only assign criticality value $1$ for the top-$u$ nodes and $0$ for all the remaining. Such criticality values can be misleading as the values may change drastically with other $u$ values. Therefore, it is necessary to consider the top-$u$ critical nodes for various values of $u$ and integrate the pieces of information, in order to have a complete view of how critical the devices are. 

In the following, we first propose an approach to find the top-$u$ critical nodes in the cellular network, via solving a bi-level mixed integer linear program (MILP). Then, based on the solution of the MILP, we discuss Alg. \ref{alg:nce}, Node Criticality Evaluation (NCE), to determine node criticality. 

\begin{center}\scriptsize
  \begin{table}[h]
		\caption{Summary of Notations for Section \ref{ssec:D2Donly}}
		\centering
    \begin{tabular}{ | l | l |}
    \hline
    Notation & Description \\		
	\hline			
		$\mathcal{D}$ & Set of mobile devices\\
		$G^c$ & $G^c = (V^c,E^c)$, the cellular network \\
		       &          with all cellular/D2D links \\ 
               $G^d$ & $G^d=(V^d,E^d)$, the D2D network \\
               $I(i,j),I(e)$ & Interference set of edge $e=(i,j)\in E^c$\\ 
		$c(i,j),c_e$ & Capacity of edge $e=(i,j)\in E^c$ \\
		$R^c,R^d$ &  Data content requesters in cellular/D2D mode\\ 
		$N^{+}(i), N^{-}(i)$  & Set of outgoing and incoming links of node $i$\\ 
        	$G^{c'}$ & $G^{c'}=(V^{c'},E^{c'})$, the modified cellular network\\
        	$E^l$ & The set of all cellular links in $G^{c'}$\\
        	$E^{d'}$ & The set of all D2D links in $G^{c'}$\\
        	$E^m$ & $E^m =\{e_i|i\in V^{d}\}$, \\
        	& the set of removable edges in $G^{c'}$\\
        	$f_{ij},f_e$ & Flow rate on edge $e=(i,j)\in E^{c'}$\\
        	$T$ & $T=f_{v^tB}$, the throughput of the network, where \\
            & $B$ denotes the BS and $V^t$ is a virtual sink node\\
        	$z_e$ & Binary variable indicates whether\\
            &attack edge $e$ ($z_e=1$) or not ($z_e=0$)\\
        	$W^c$ & Bandwidth allocated to cellular transmission\\
        	$W^d$ & Bandwidth allocated to D2D transmission\\
        $k$ & Budget of the adversary\\
			
    \hline
    \end{tabular}		
		\label{table:MFINotations}
	\end{table}
\end{center}
                
\subsubsection{Find the Critical Nodes in the Cellular Network}

Under the system model discussed in \ref{ssec:cellularnetwork}, the cellular network can be modeled as a flow network. The capacity and flow of a cellular/D2D link is its maximum data rate and actual data rate, respectively. However, special care is required in the construction, for linearizing the calculation of data rate, and modeling devices switching from cellular to D2D mode, which will be discussed respectively as follows. 

\textbf{Data Rate Estimation.}  In practice, calculating data rate often requires nonlinear formulas considering interference management and resource allocation. However, adding nonlinear constraints to a formulation can greatly increase its complexity. To avoid this, we describe an approach to estimate data rate. As the main purpose of studying TMIN is to demonstrate how cellular network throughput can be impacted by rumors in OSNs, the estimated data rates suffice. 

To estimate data rates, we first discuss the collision management for D2D and cellular links. Links sharing the same band cannot transmit concurrently if they are in the same access domain, in order to avoid collision. We denote such links as an Interference Set. The interference set for a link $e=(i,j)$, $I(e)$, can be defined as $I_e=\{(i',j')|d_{ii'}\leq \beta d_0\}$ where $\beta$ is a tunable parameter and $d_0$ is the maximum D2D communication distance. Notice that the link $e$ itself is also contained in $I(e)$. With the interference sets for all links, we can obtain the set of links that are allowed to transmit at the same time, $L(e)$, given one active link $e$. $L(e)$ is constructed as follows. First we set $L(e)=\{e\}$. Then we iteratively select a link into $L(e)$ if it causes the largest drop in data rate among links that are not in the interference sets of links in $L(e)$. Assume each device has a fixed transmission power, the SINR for each D2D link can therefore be obtained using \eqref{eqn:sinr}. Based on \eqref{eqn:d2drate}, the data rate is proportional to the bandwidth, which is a variable. For notation convenience, we calculate $c(e)$ using \eqref{eqn:d2drate}, assuming unit bandwidth. Therefore, the maximum data rate of link $e$ can be expressed by $W^d\times c(e)$, recall that $W^d$ is the bandwidth assigned to D2D communications. Similarly, we can use \eqref{eqn:cellularrate}, \eqref{eqn:snr} to obtain the maximum data rate for a cellular link $e'$, which is $W^c\times c(e')$. The actual data rate $f(e)$ is modeled as a linear variable upper bounded by the maximum data rate. 

\textbf{Network Modification}
To create a flow network with a single source and a single sink, we introduce additional components, including a global sink node $v^t$ as well as the links $(i,v^t), \forall i\in R$ and $(v^t,B)$. A link among those has infinite capacity and an empty interference set. Next, we model nodes switching from D2D mode to cellular mode. Two types of devices may switch their modes. The first type is the relay devices $V^r$. When such devices turn to cellular mode, we can remove them from the network as they neither request any data nor contribute to D2D transmission. The second type is the receiver devices $R^d$. When such devices switch their modes, they still request data from the BS. So we must remove all D2D links associated with those devices, but keep their cellular links. To model a switch as a link removal, we split each node $i\in V^d$ into two nodes $i^{-},i^{+}$. The incoming D2D links are connected to $i^-$ and all outgoing links are connected to $i^+$. Denote 
$E^{d'} = \{(i^+,j^-)|(i,j)\in V^d\}$ as the set of all D2D links. For the cellular link $(B,i)$, we replace it with $(B,i^-)$ if $i\in V^r$ and $(B,i^+)$ if $i\in R^d$. Let $E^{l} = \{(B,i)|i\in v^c\backslash v^d\}\cup\{(B,i^-)|i\in V^r\}\cup \{(B,i^+)|i\in R^d\}$ as the set of all cellular links. Nodes $i^-,i^+$ are connected by link $e_i=(i^{-},i^{+})$, which has an empty interference set and infinite capacity. Let $E^m=\{e_i|i\in V^{d}\}$. Denote the modified graph as $G^{c'}=(V^{c'},E^{c'})$ where 
\begin{align*}
V^{c'}=(V^c\backslash V^d) \cup \{i^-,i^+|i\in V^d\} \cup \{B,v^t\}
\end{align*}
and
\begin{align*}
E^{c'}= E^l\cup E^{d'}\cup E^m\cup\{(v^t,B)\}\cup \{(i,v^t)|i\in R\}
\end{align*}
We can easily verify that switching a node $i$ from D2D to cellular in $G^c$ is equivalent to removing link $(i^-,i^+)$ in $G^{c'}$. 

With the data rate model and the modified network $G^{c'}$, we formulate the problem as a bi-level MILP $\mathcal{P}$.
\begin{align}\scriptsize
\mathcal{P}:	\min T(\textbf{z})&\label{eq1:outerobj}\\
	s.t.  \sum_{e\in E^m} z_e\leq u\label{eq1:resource}&\\
		z_e\in\{0,1\},& \forall e\in E^m\\
	T(\textbf{z})=\max f_{v^tB}\label{eq1:innerobj}&\\
		s.t.\quad   \sum_{i\in N^{-}(j)}f_{ij}-\sum_{k\in N^{+}(j)}f_{jk}&= 0, \forall j\in V^{c'}\label{eq1:flowbalance}\\
             \sum_{e'\in I(e)}\frac{f_{e'}}{W^c\times c(e')}&\leq 1, \forall e\in E^l\label{eq1:interferencecellular}\\
            \sum_{e'\in I(e)}\frac{f_{e'}}{W^d\times c(e')}&\leq 1, \forall e\in E^{d'}\label{eq1:interferenced2d}\\
			 f_e\leq c(e)(1-z_e),&\quad \forall e\in E^m \label{eq1:removal}\\
             W^c+W^d\leq W \label{eq1:bandwidth}\\
			 f_{e}\geq 0,&\quad \forall e\in E^{c'}\\
             W^c,W^d \geq 0
\end{align}
The objectives \eqref{eq1:outerobj} (outer stage) and \eqref{eq1:innerobj} (inner stage) guarantees the problem to be a minimization of the maximum throughput $T$. Constraint \eqref{eq1:resource} restrict the solution to be top-$u$ critical nodes. The binary variable $z_e$ reaches $1$ if link $e$ is removed. Therefore, at most $u$ links in $E^m$, i.e. at most $u$ devices in $V^d$ can be removed. Constraint \eqref{eq1:flowbalance} is the flow balance constraint. The traffic received by a device must be equal to what it transmits. Collision management in the cellular network are considered in \eqref{eq1:interferencecellular} and \eqref{eq1:interferenced2d}, by which we ensure that only one link among all links in an interference set can transmit at any point of time. Also, since the traffic $f_e$ must be non-negative, \eqref{eq1:interferencecellular} and \eqref{eq1:interferenced2d} upperbound the traffic on any link by its data rate. Constraint \eqref{eq1:removal} model the case that no flow can be assigned to a removed edge. Notice that we omit all links with infinite capacity for the capacity constraints. Constraint \eqref{eq1:bandwidth} limits the total bandwidth being used by cellular and D2D communications. 

As $\mathcal{P}$ cannot be solved directly due to its bi-level structure, we reformulate it as a single-level MILP $\mathcal{P}^d$ by dualization and linearization techniques. The solution of $\mathcal{P}^d$ is attainable via existing MILP solvers.

\textbf{Dualization of the Inner Stage.}
By dualization, we can transform the inner stage to a equivalent minimization problem, so that the original minimax bi-level formulation is equivalent to a minimization problem in only one stage. Denote $p,q_c,q_d, r, l$ as the dual variables corresponds to constraints \eqref{eq1:flowbalance}-\eqref{eq1:bandwidth}, respectively. The reformulated program after dualization is denoted as $\mathcal{P}^d$. 
\begin{align}\scriptsize
	\mathcal{P}^d: \min  W\times l + \sum_{e\in E^m}  c(e)(r_e-\delta_e)&&\\
	 s.t. \sum_{e\in E^m} z_e\leq u&&\\
     \delta_e\leq z_e,&& \forall e\in E^m \label{eqn:linear1}\\
     \delta_e\leq r_e,&& \forall e\in E^m\label{eqn:linear2}\\
     \delta_e \geq r_e-(1-z_e),&& \forall e\in E^m\label{eqn:linear3}\\
	 p_B-p_{v^t}\geq 1 &&\label{eqn:pbvt}\\
     p_{v^t}-p_i \geq 0, && \forall i\in R\label{eqn:pvr}\\
	 p_j-p_i+\sum_{e\in I(i,j)}\frac{q_{e}^c}{c(i,j)}\geq 0,&&  \forall (i,j)\in E^{l}\label{eqn:pcellular}\\
     	 p_j-p_i+\sum_{e\in I(i,j)}\frac{q_{e}^d}{c(i,j)}\geq 0,&&  \forall (i,j)\in E^{d'}\label{eqn:pd2d}\\
	 p_j-p_i+r_e\geq 0,  \forall e=(i,j)\in E^{m}&&\label{eqn:constforr}\\
     l - \sum_{e\in E^l}q^{c}_{e}\geq 0 &&\label{eqn:l1}\\
     l - \sum_{e\in E^{d'}}q^{d}_{e}\geq 0 \label{eqn:l2}&&\\
	 z_e\in\{0,1\},r_e\geq 0, \delta_e\geq 0, && \forall e\in E^m\\
     q^{c}_{e}\geq 0, && \forall e\in E^l\\
     q^{d}_{e}\geq 0, && \forall e\in E^{d'}\\
	 l\geq 0 &&
\end{align}
\textbf{Linearization.} The original dual objective contains a quadratic term $\sum_{e\in E^m} c(e)r_e(1-z_e)$, which largely increases complexity. To achieve a linear formulation, we substitute $r_ez_e$  with $\delta_e$ and add constraints \eqref{eqn:linear1} to \eqref{eqn:linear3}. The equivalence after the step is proved in Lemma \ref{lemma:linearization}. 
\begin{lemma}\label{lemma:linearization}
	\eqref{eqn:linear1} to \eqref{eqn:linear3} guarantee $\delta_e=r_ez_e, \forall e\in E^m$.
\end{lemma}
\begin{proof}
First, we claim that the range of any $r_e$ is $[0,1]$. Based on the objective, both $l$ and $r_e$ should be minimized. As the only constraint that potentially requests an $r_e$ value be higher than $0$ is inequality \eqref{eqn:constforr}, the value of $r_e$ can be written as $\max \{0, p_i-p_j\}$ for $e=(i,j)\in E^m$. Based on constraints \eqref{eqn:pbvt} to \eqref{eqn:pd2d}, the largest value for any $p_i-p_j$ is 1, otherwise the $l$ value will be unnecessarily large and the solution is not optimal. Therefore, $0\leq r_e\leq 1$ for all $e\in E^m$. By constraints \eqref{eqn:linear1} to \eqref{eqn:linear3}, when $z_e=0$, $\delta_e=0=z_er_e$; when $z_e=1$, $\delta_e=r_e=z_er_e$(as $r_e\leq 1=z_e$). As $z_e$ is binary, we have $\delta_e=z_er_e$ in all cases.
\end{proof}

\subsubsection{The NCE Algorithm}
By the discussion at the beginning of Sect.~\ref{ssec:D2Donly}, we can use the solutions of $\mathcal{P}^{d}$ to support criticality evaluation as in Alg.~\ref{alg:nce}. The program $\mathcal{P}^{d}$ is solved under a set $\mathcal{U}=\{u_1,u_2,...u_w\}$ of budgets. Denote the solutions as $\textbf{z}^1$,...,$\textbf{z}^w$, we define the criticality of node $i$ as
\begin{align*}
cr_i=\sum_{p=1,...,w}z^{p}_{e_i}
\end{align*}
The definition of $cr_i$ is aligned with the concept of being critical. A device $i$ that appears in more $\mathcal{P}^{d}$ solutions has higher criticality than another device $j$ that contributes to less $\mathcal{P}^{d}$ solutions. Intuitively, all users with no devices in the considered cellular network have $0$ criticality. 

\begin{algorithm}
 	\caption{\textbf{Node criticality Evaluation(NCE)}}
    \label{alg:nce}
    \begin{algorithmic}
        \REQUIRE Cellular network $G^{c}=(V^{c},E^{c})$, List of budgets $\mathcal{U}=\{u_1,u_2,...,u_w\}$
        \ENSURE $\mathcal{C}=\{cr_i| \forall i \in V^d\}$
     	\STATE Initialize $cr_i=0, \forall i\in V^{d}$\\
        \FORALL{$u_p\in \mathcal{U}$}
        
        \STATE Solve $\mathcal{P}^{d}$ with budget $u_p$ in $G^{c}$.\\
        \STATE Denote the solution as \textbf{z}.\\
        \STATE $cr_i+=z_{e_i}, \forall i\in V^d$.
        \ENDFOR 
     	\RETURN $\mathcal{C}$
    \end{algorithmic}
 \end{algorithm}

We assume the availability of the full knowledge of the cellular network in this paper. In reality, a BS covering a business region may observe some frequently reappearing location patterns or mobility traces of the UEs in the form of snapshots, or static networks. Therefore, if an adversary wants to exploit the vulnerability of the interdependent cellular and social networks, s/he can calculate the criticality of the devices using Alg.~\ref{alg:nce} in all snapshots. However, how an adversary attacks the networks in reality is out of the scope of this paper, as our aim is to understand the vulnerability in the worst case.

\subsection{Targeted IM Algorithm}\label{ssec:bothnetwork}
With Alg. \ref{alg:nce}, we can quantitatively evaluate the criticality of all nodes in the OSN. Due to the way the criticality values are assigned, when more nodes with large criticality values in the OSN are influenced, we can expect a more severe throughput reduction in the cellular network. Therefore, the problem of finding the top-$k$ critical nodes in the OSN can be interpreted as the targeted-IM problem of finding $k$ seed nodes to maximize the total criticality of all influenced nodes. In the following, we first propose an efficient Targeted-IM algorithm based on the reverse influence sampling (RIS) technique and then prove its guaranteed $(1-\frac{1}{e}-\epsilon)$ approximation ratio and near-linear running time.

\begin{center}
  \begin{table}[h]
		\caption{Summary of Notations for Section \ref{ssec:bothnetwork}}
		\centering
    \begin{tabular}{ | l | l |}
    \hline
    Notation & Description \\		
	\hline			
		$G^s$ & $G^s=(V^s,E^s)$, the social network \\ 
        	$\mathcal{C}$ & $\mathcal{C}=\{cr_i| \forall i \in V^s\}$, \\
        	& criticality of all users \\ 
		$\mathbb{I}^c(S)$ & Expected sum of criticality over all the \\
		& nodes influenced by $S$  \\ 
		$S^{*}$ & The optimal seed set \\ 
        	$\mathcal{R}$ & The collection of all RR sets.\\
	 $\deg_{\mathcal{R}}(S)$  & \#  of hyperedges incident to nodes in $S$ \\& among $\mathcal{R}$\\ 
        $\Omega$ & $\Omega = \sum_{v\in V^s} cr_v$\\
        $\bar{\mathbb{I}}^c(S)$ & Estimator for $\mathbb{I}^c(S)$. $\bar{\mathbb{I}}^c(S)=\frac{\deg_{\mathcal{R}}(S)}{|\mathcal{R}|}\Omega$ \\
        	$\epsilon,\delta$ & Precision parameters of Alg. \ref{alg:tim}\\
		$\tau$ & $\tau=\sqrt{\ln(\frac{2}{\delta})}$\\
		$\sigma$ & $\sigma=\sqrt{(1-\frac{1}{e})(\ln{\binom{n}{k}+\ln\frac{2}{\delta}})}$\\
		$\phi$ & $\phi = \frac{(1-1/e)\sigma+\tau}{\epsilon}$ \\
		$\gamma$ & Stopping criterion for generating RR \\
		& sets, $\gamma=2(\phi^2 + \log \frac{2}{\delta})$ \\  
		$\bar{S}$ & Seed set outputted by Alg. \ref{alg:tim}\\
			
    \hline
    \end{tabular}		
		\label{table:bothnetwork}
	\end{table}
\end{center}

\textbf{Brief Review of the RIS Technique.} RIS was first proposed in \cite{borgs14reverse}, which casts light on efficient IM algorithms. It first samples Reverse Reachable (RR) sets and then apply a greedy maximum coverage (MC) algorithm to obtain the seed set. An RR set $\tilde{R}$ consists of the set of nodes that can reach an origin node $o_{\tilde{R}}$ in a sample graph $\tilde{G}$. With enough RR sets, one can estimate the influence each seed set can have to the whole network, by relating the seed set and the RR sets to a coverage instance, in which an RR set corresponds to an element and a node in the graph corresponds to a set. An element is covered by a set if and only if the node exists in the RR set. Approximately, the influence spread of a seed set is positively correlated to the number of covered RR sets over the total number of RR sets, greedily solving the MC can output a seed set with near-optimal influence spread. However, to establish an accurate estimation of influence spread of each seed set without exert too much burden to computation, the number of RR sets generated must be carefully selected. 

In our scenario, as the nodes are of different criticality values, the probability of starting a random RR set from node is $\frac{cr_{v}}{\Omega}$, proportional to its criticality values. Denote $\deg_{\mathcal{R}}(S)$ as the number of RR sets in $\mathcal{R}$ covered by seed set $S$. 

To limit the number of RR sets, we set exponential check points during RR set generation. At each check point, we greedily solve the weighted MC problem \cite{vazirani2013approximation} and examine if $\deg_{\mathcal{R}}(\bar{S})$ exceeds a threshold, where $\bar{S}$ is the solution by the greedy MC. If so, we stop RR set generation and output $\bar{S}$ as the seed set. The Targeted-IM algorithm is presented in Alg. \ref{alg:tim}. With an improved threshold, the time complexity of Targeted-IM is a constant smaller than the state-of-art target IM algorithm BCT \cite{nguyen2016targeted}.

 \begin{algorithm}
 	\caption{\textbf{Targeted-IM}}
    \label{alg:tim}
    \begin{algorithmic}
        \REQUIRE	Social network $G^s=(V^s,E^s)$,with criticality values on nodes, seed set size $k$, Precision parameters $\epsilon>0$, $\delta\in(0,1)$, $\gamma$
        \ENSURE Seed set $S\subseteq V^s$
       	\STATE Collection of RR sets $\mathcal{R} = \emptyset$
        \STATE $N_R=\gamma,N_{R}^{0}=1$
     	\WHILE{$\deg_{\mathcal{R}}(\bar{S})<\gamma$}
        	\FOR{$i=N_{R}^{0}$ to $N_R$}
           		\STATE Generate RR set $\tilde{R}$ by BSA in \cite{nguyen2016targeted}.
    			\STATE	$\mathcal{R}=\mathcal{R}\cup \{\tilde{R}\}$
            \ENDFOR
            \STATE $\bar{S}=GreedyMC(\mathcal{R},k)$
            \STATE　$N_{R}^{0}=N_R,N_R=2N_R$
     	\ENDWHILE
     	\RETURN $\bar{S}$
    \end{algorithmic}
 \end{algorithm}

In the following, we present the approximation ratio and time complexity of Alg. \ref{alg:tim}.  For conciseness, the proof for Theorem \ref{theorem:timratio} is placed in appendix. 
\begin{theorem}\label{theorem:timratio}
	With $\epsilon > 0$, $0<\delta < 1$, the seed set $\bar{S}$ calculated by Alg. \ref{alg:tim} satisfies
		$\mathbb{I}^c(\bar{S})\geq (1-\frac{1}{e}-\epsilon)\mathbb{I}^c(S^*)$
	with probability at least $1-\delta$.
\end{theorem}

\begin{theorem}\label{theorem:timtime}
	The expected running time of Alg. \ref{alg:tim} is $O(\gamma|E^s|)$.
\end{theorem}
The proof is a straightforward extension of the one in \cite{nguyen2016targeted}, and is omitted.

\section{Experiments}\label{sec:experiment}

In this section, we evaluate how cellular network performance can be impacted when the critical nodes in the OSN propagates the rumor. In the following, we first discuss the setup for the experiments in Section \ref{ssec:setup}. In Section \ref{ssec:misinfo}, we measure the vulnerability of the cellular network. Then we consider how the actions of the users and the BS can reduce the vulnerability in Section \ref{ssec:awareness} and \ref{ssec: retention}, respectively.

\subsection{Setup}\label{ssec:setup} 

\textbf{Cellular Network.} As the real data sets are limited due to our requirements, we choose to use synthetic data for the cellular network. In a square $50$m by $50$m area, we generate $50$ data requesters at random locations. Further, we select $30$ of them to be in D2D mode and the other $20$ remain in cellular mode. To enable D2D communication, we also generate $60$ D2D relay devices at random locations within the area.  All the wireless parameters are summarized in Table~\ref{table:wirelessparam}. 


\textbf{OSN.} To model information propagation in the OSN, we use real-world Facebook network topology from \cite{snap}. The Facebook data consists of 4039 users and 88234 social network interactions between the users \cite{mcauley2012learning}. The influence between a pair of users was assigned randomly from a uniform distribution.
\begin{table}
	\caption{Wireless Network Parameters} 
	\begin{tabular}{|c|c|}
		\hline Notation & Description \\
        \hline Cell dimension & $50$ x $50$ $m^2$ \\ 
        \hline BS position  & Lower left corner \\ 
        \hline Network bandwidth & $0.1$ MHz \\ 
		\hline Channel Model & Multipath Rayleigh
fading \\ 
        \hline Path Loss Exponent & $3$ \\ 
        \hline Noise spectral density & $-174$ dBm/Hz \\ 
		\hline BS transmit power& $100$ W \\ 
		\hline D2D transmit power & $10$ W \\ 
        \hline D2D distance & $15$ m\\
		\hline
	\end{tabular}

	\label{table:wirelessparam}
\end{table}

\textbf{Interdependency.} For the interconnection between the two networks, we consider two scenarios, namely a stadium scenario and a shopping mall scenario. In the stadium scenario, the owners of the devices are not very likely to be socially connected. In the shopping mall scenario, however, as the shoppers normally resides closer than the game goers, the owners of the devices are more likely to be connected in the OSN, compared with the stadium scenario. 

We use the following procedure to create the interconnections. For each D2D node $i\in V^d$, we construct link $e^{vi}$ between $i$ and a uniformly randomly chosen user $v\in V^s$ from the OSN. Once $e^{vi}$ is constructed, with probability $p_1$, we uniformly randomly choose a neighbor $i'$ of $i$ and a neighbor $v'$ of $v$ and construct link $e^{v'i'}$. With probability $p_2$, we construct link $e^{\bar{v}\bar{i}}$ for a random neighbor $\bar{i}$ of $i'$ and a random neighbor $\bar{v}$ of $v'$, when $e^{v'i'}$ is constructed. In order to differentiate the different level of social connections among users in the cell area, we apply higher $p_1,p_2$ values for the shopping mall scenario. In our experiments we have set $p_1=0.7$ and $p_2=0.4$ for the stadium scenario, and $p_1=0.9$ and $p_2=0.6$ for the shopping mall scenario.

\textbf{Algorithms.} To illustrate the efficacy of our proposed method RCF, especially the criticalness evaluation scheme, we compare RCF with two approaches. The first approach uses degree centrality as criticalness values and runs the same targeted IM algorithm embedded in RCF to obtain the critical nodes in the OSN. The second approach, as a baseline, randomly assigns criticalness values to guide critical nodes selection in the OSN. 


\subsection{Vulnerability of the Cellular Network}\label{ssec:misinfo}
In this section, we study how vulnerable the cellular network is to rumors in the OSN, in terms of throughput reduction. Also, we demonstrate that social connectivity of the nodes in the cellular network can impact the vulnerability. Additionally, we explore another dimension of vulnerability: how much bandwidth is necessary to bring the throughput back to the original value. 
\begin{figure}[!ht] \centering
\caption{Throughput Reduction}
\subfloat[Stadium Scenario]{
 	\includegraphics[width=0.49\linewidth]{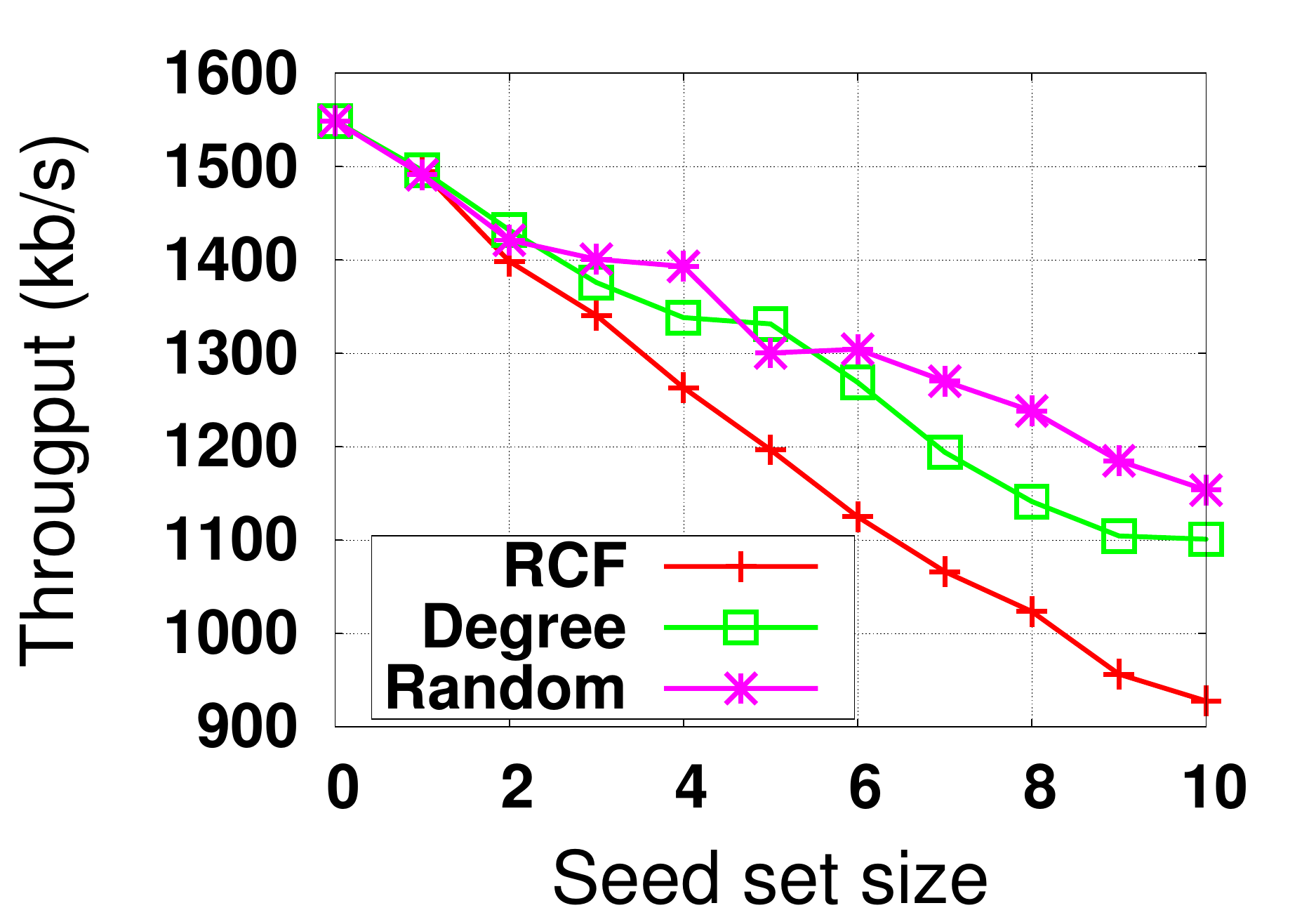}
 	\label{fig:thr}}
\subfloat[Shopping Mall Scenario]{
 	\includegraphics[width=0.49\linewidth]{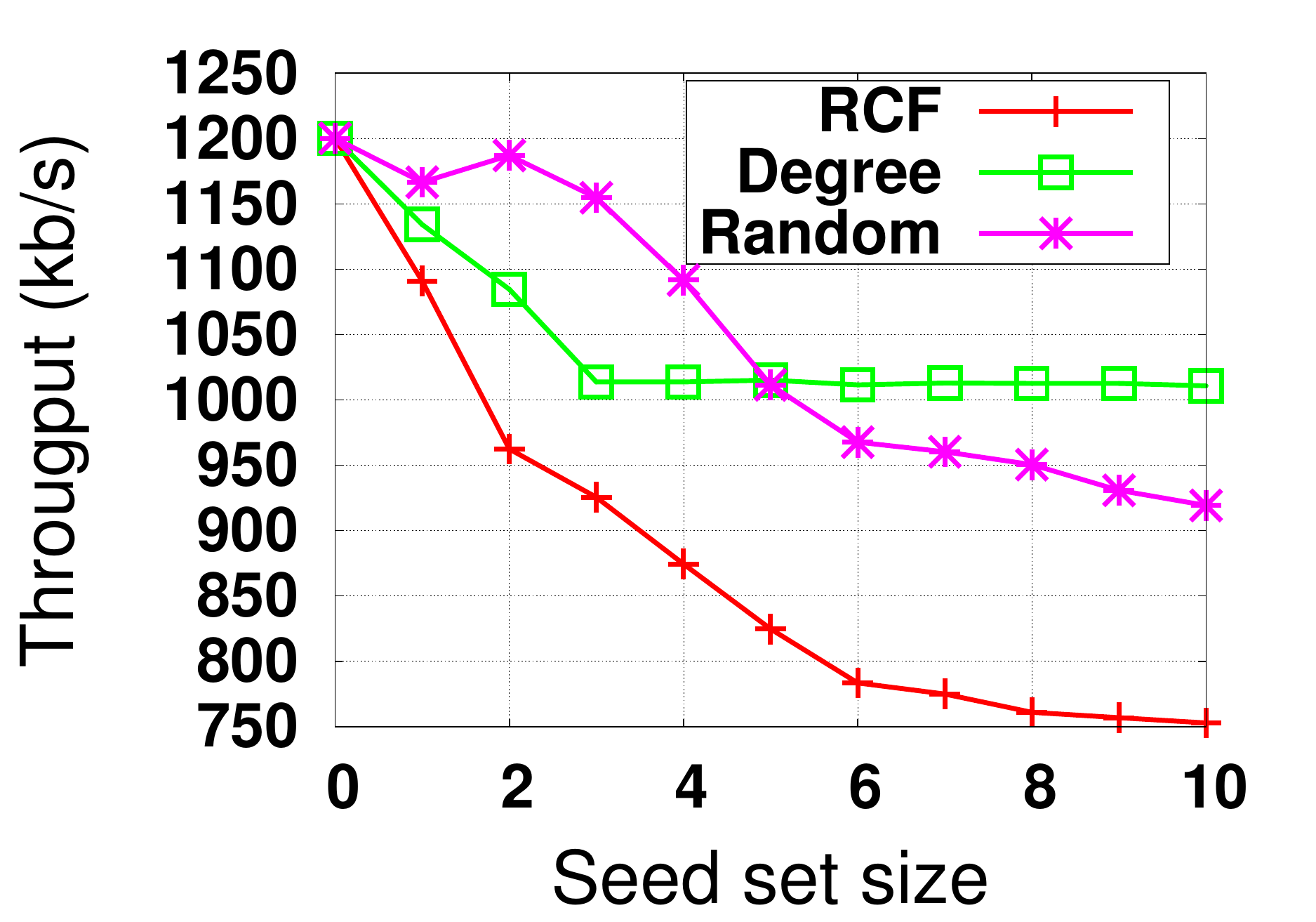}
 	\label{fig:user_affected}}

    \label{fig:throughputdec}
 \end{figure}
 
We first observe how the throughput can be impacted with varying seed set size of spreading rumors in the OSN. In Figure \ref{fig:throughputdec}, for both scenarios, the throughput of the cellular network experiences substantial decreases after the rumor propagation. The decrements are higher with increasing number of seeds (i.e. critical nodes). With RCF, the damage to throughput is severer than other approaches, which means that RCF can better reveal the critical nodes. The gap between RCF and the other methods increases significantly as more seed nodes can be selected. In Figure \ref{fig:thr}, the throughput drops $40\%$ when $10$ critical nodes outputted by RCF are used as seeds of rumor. It is worth noticing that even the random approach can lead to noticeable decrease in throughput. Therefore, the interdependent cellular/social networks is vulnerable not only to rumors that are intentionally spread, but also to those that are randomly originated.

\begin{figure*}[!ht] 
\begin{minipage}[t]{0.31\linewidth}
\centering
\includegraphics[width=1\linewidth]{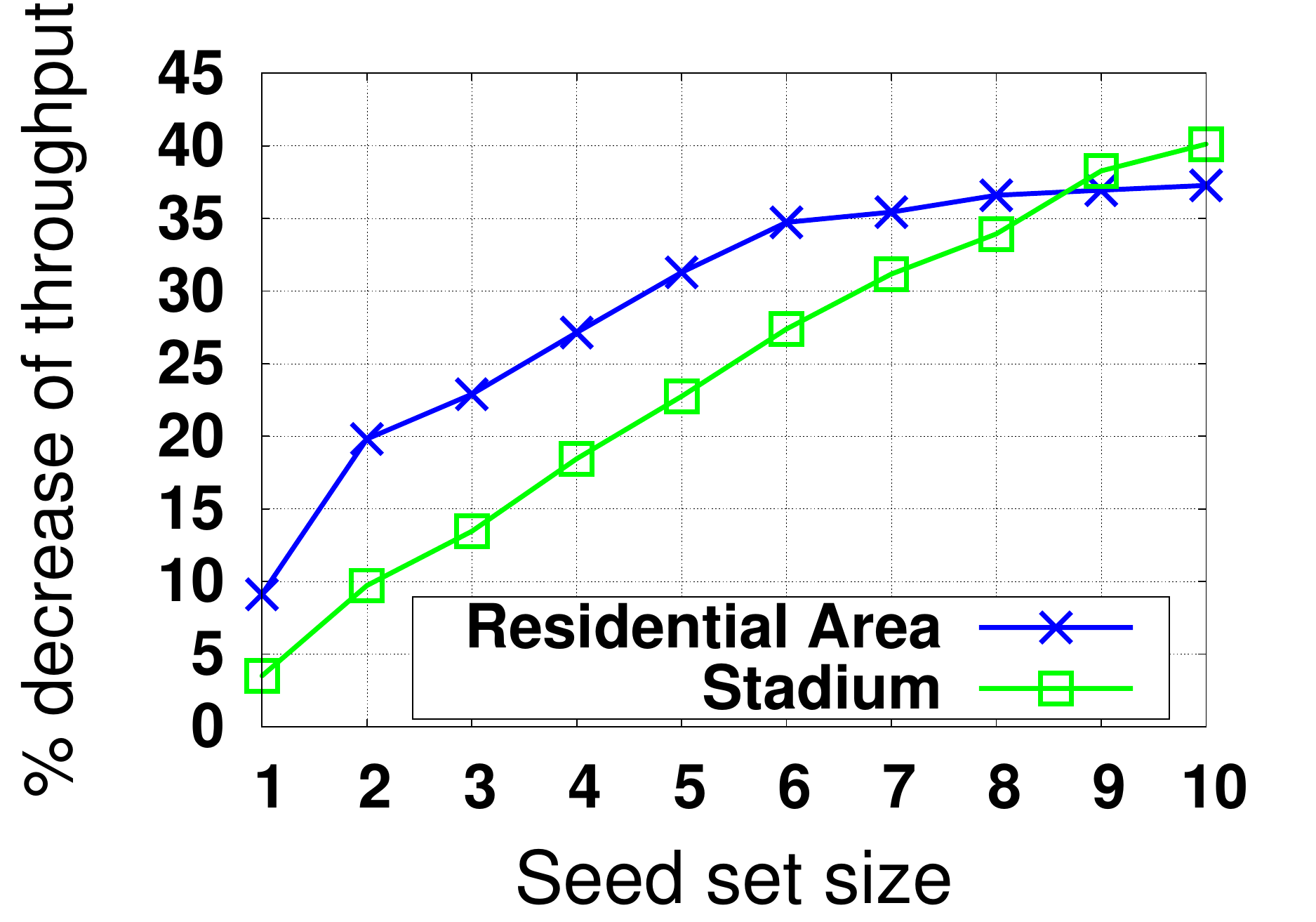}
\caption{Stadium versus Shopping Mall: higher social connectivity results in faster drop of throughput}
\label{fig:res}
\end{minipage}
\hfill
\begin{minipage}[t]{0.31\linewidth}
\centering
 	\includegraphics[width=1\linewidth]{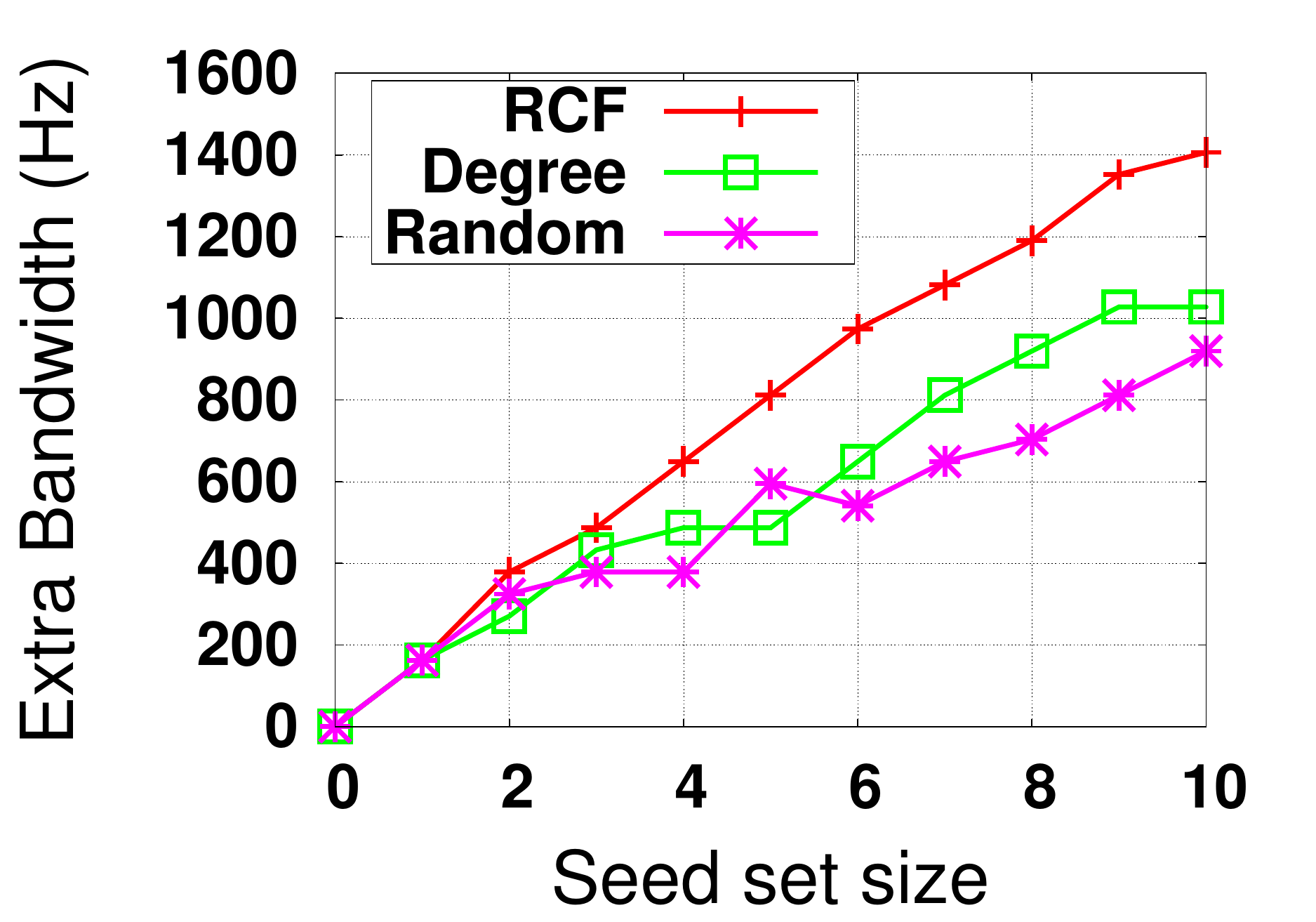}
 	\caption{Extra bandwidth requirement for the stadium scenario}
 	\label{fig:new_resource}
\end{minipage}
\hfill
\begin{minipage}[t]{0.31\linewidth}
\centering
 	\includegraphics[width=1\linewidth]{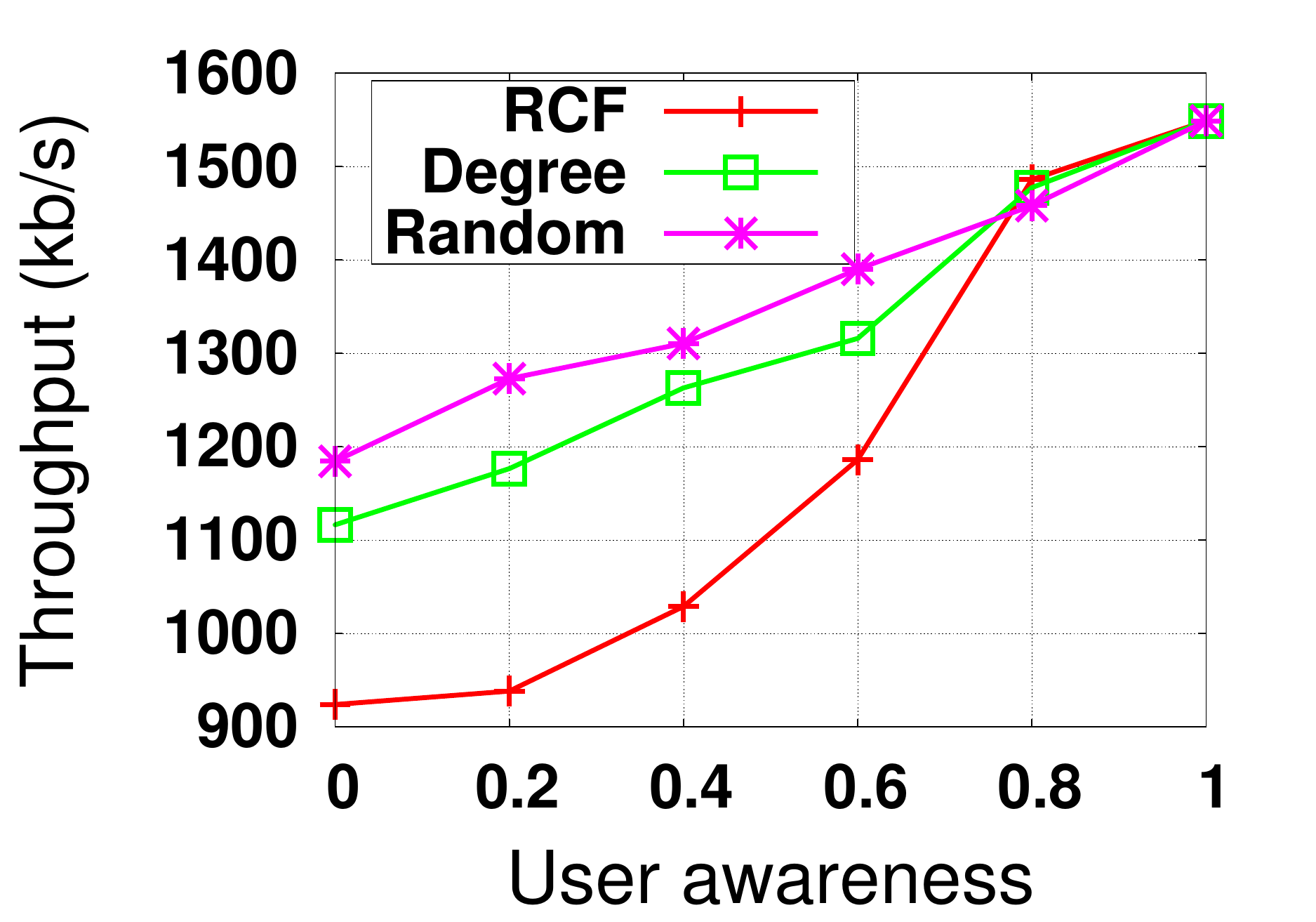}
 	\caption{Impact of user awareness on network throughput (adversary's budget is 10)}
 	\label{fig:ua_k10}
 \end{minipage}

\end{figure*}

Next, we examine how social connectivity impacts the throughput. In theory, in a network with higher level of social connections, rumors can cause larger influence than a network with less social connections, with the same seed set size. Therefore, the cellular network can be more vulnerable when the interdependent OSN is more densely connected. 

The result depicted in Figure \ref{fig:res} supports the analysis. At each seed set size,  we calculate the relative throughput decrease $Q_k$ for different scenarios. Denote the throughput after rumor propagation with $k$ seed nodes as $T_k$, we define $Q_k=\frac{T_0-T_k}{T_0}\times 100\%$. 

As can be seen from Figure \ref{fig:res}, $Q_k$ for the shopping mall scenario increases faster in most of the cases, which demonstrates that the social connections actually aid the rumor propagation, as well as increase the vulnerability. The slow increase in $Q_k$ for the shopping mall scenario after seed set size $8$ can be interpreted as a saturation case, that all D2D users are already influenced and the throughput cannot be further decreased. It explains why $Q_k$ of the stadium scenario surpasses that of the shopping mall scenario after seed set size $8$.

Then, we show in Figure \ref{fig:new_resource} the additional bandwidth cost requirement that arises from increasing the throughput to the level before rumor propagation. Consistent with our previous findings, the critical nodes found by RCF result in the largest bandwidth requirement and thus exerts the highest pressure to the BS. After the rumor propagation, the extra bandwidth required can be up to $1400$Hz. When a BS runs out of its own bandwidth and fails to supply the additional bandwidth, the throughput will be decreased drastically and the data requesters can experience a low data rate, which is not preferable. However, supply the extra requirement can be costly and thus a heavy burden to the BSs. This phenomenon emphasizes how vulnerable the cellular network can be, from another direction.  

\subsection{Reduce Vulnerability from the OSN}\label{ssec:awareness}
In this section, we study how users' actions can help in reducing the vulnerability of the cellular network. In OSNs, not all users are equally sensitive to rumors of D2D communication. Users having larger exposure to related knowledge are likely to identify the rumors and remain uninfluenced. So, they are not likely to opt out using D2D nor spread the rumor. In contrast, the users with limited knowledge to D2D are susceptible to rumors and are much more likely to opt out using D2D.  

To quantitatively describe the users' knowledge to D2D, we use the term "User Awareness" (UA) which ranges from $0.0$ to $1.0$ with $1.0$ means most aware of D2D related rumors. UA is in effect in the IC model in the following way. Denote the awareness for user $v$ as $W_v$. In the original IC model, node $u$ can influence node $v$ with probability $p_{uv}$. With UA, the probability is modified to $p_{uv}\times(1-W_v)$. Therefore, the most aware users ($W_v=1$) can never be influenced. Clearly, when users become more aware of the rumors, the rumor propagation will be limited and the vulnerability of the cellular network to the rumors will be reduced. 

In this experiment, we fix the number of seed nodes at $10$ and vary UA values between $0$ and $1$ in the stadium scenario. As the results for the shopping mall scenario have the similar trend, we choose not to include them for conciseness. Notice that we set the UA values equally for all nodes. Such a case can describe the situation in which the general public is educated about D2D and the individuals have similar awareness. Figure \ref{fig:ua_k10} clearly illustrates the trend that when the user are more aware, the cellular network is less vulnerable in terms of throughput. For the degree based and random methods, the throughput recovered linearly with respect to the increase of UA. However, the seed nodes selected by RCF are "resistant" to UA. When the UA value is $0.2$, the throughput after the RCF based rumor propagation only increased by $5\%$ compared with the case without UA. This experiment illustrates that the cellular network can still be vulnerable in the worst case, even if users start to be aware of the rumors.   

To have a thorough understanding of UA, we now present the network throughput achieved by RCF for different values of user awareness values under various seed set size. Figure \ref{fig:heatmap} is the heat map representing the impact of user awareness and the number of critical nodes on the network throughput. The throughput is depicted by the RGB colors. As the throughput value gets higher, the color becomes lighter in the heat map. From Figure \ref{fig:hm}, we can see that the color is the lightest, i.e. the throughput is the highest, in the top left corner where the number of critical nodes is the least and the users are the most aware. 

\begin{figure*}[!htbp]
  \centering
\begin{minipage}[c]{0.66\textwidth}
\centering
\subfloat[][The Heat Map]{\includegraphics[width=0.49\textwidth]{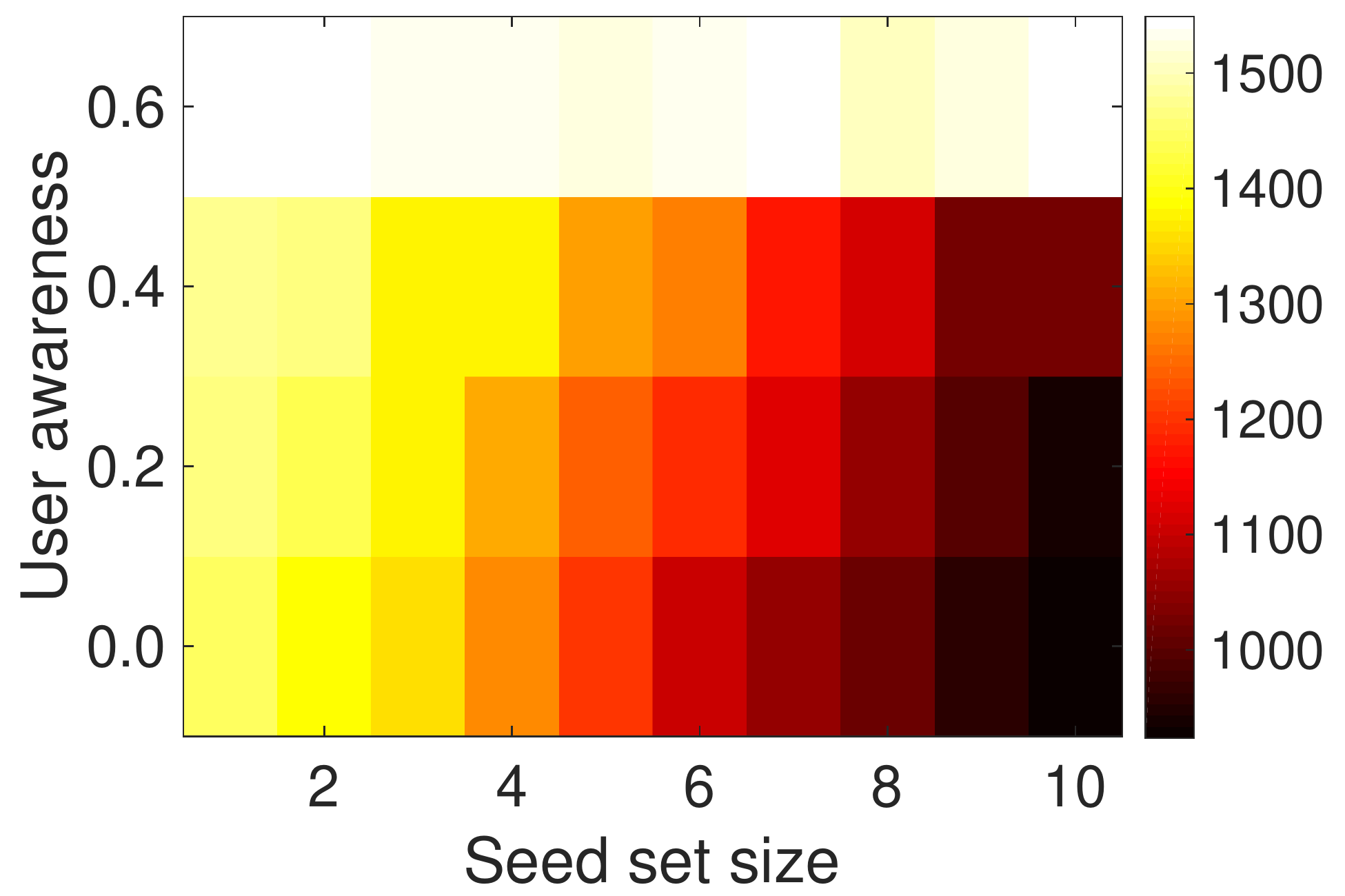}\label{fig:hm}}
~
\subfloat[][A snap shot from Figure \subref{fig:hm}: Effect of the seed set size on the throughput.]{\includegraphics[width=0.49\textwidth]{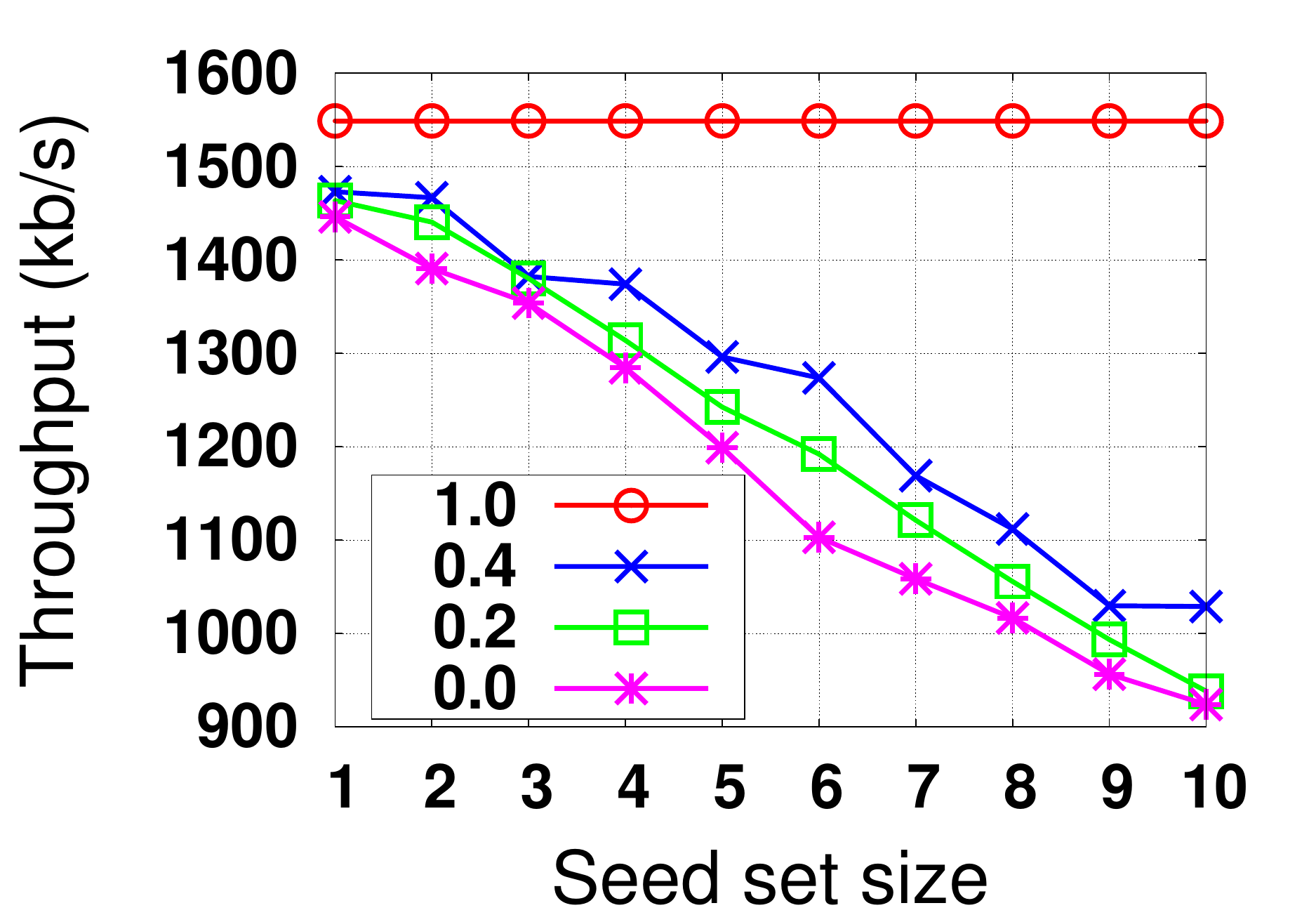}\label{fig:snap_shot}}
~

\caption{\label{fig:heatmap} Impact of user awareness on network throughput}
\end{minipage}%
\hfill
\begin{minipage}[c]{0.33\textwidth}
\centering
\includegraphics[width=1\textwidth]{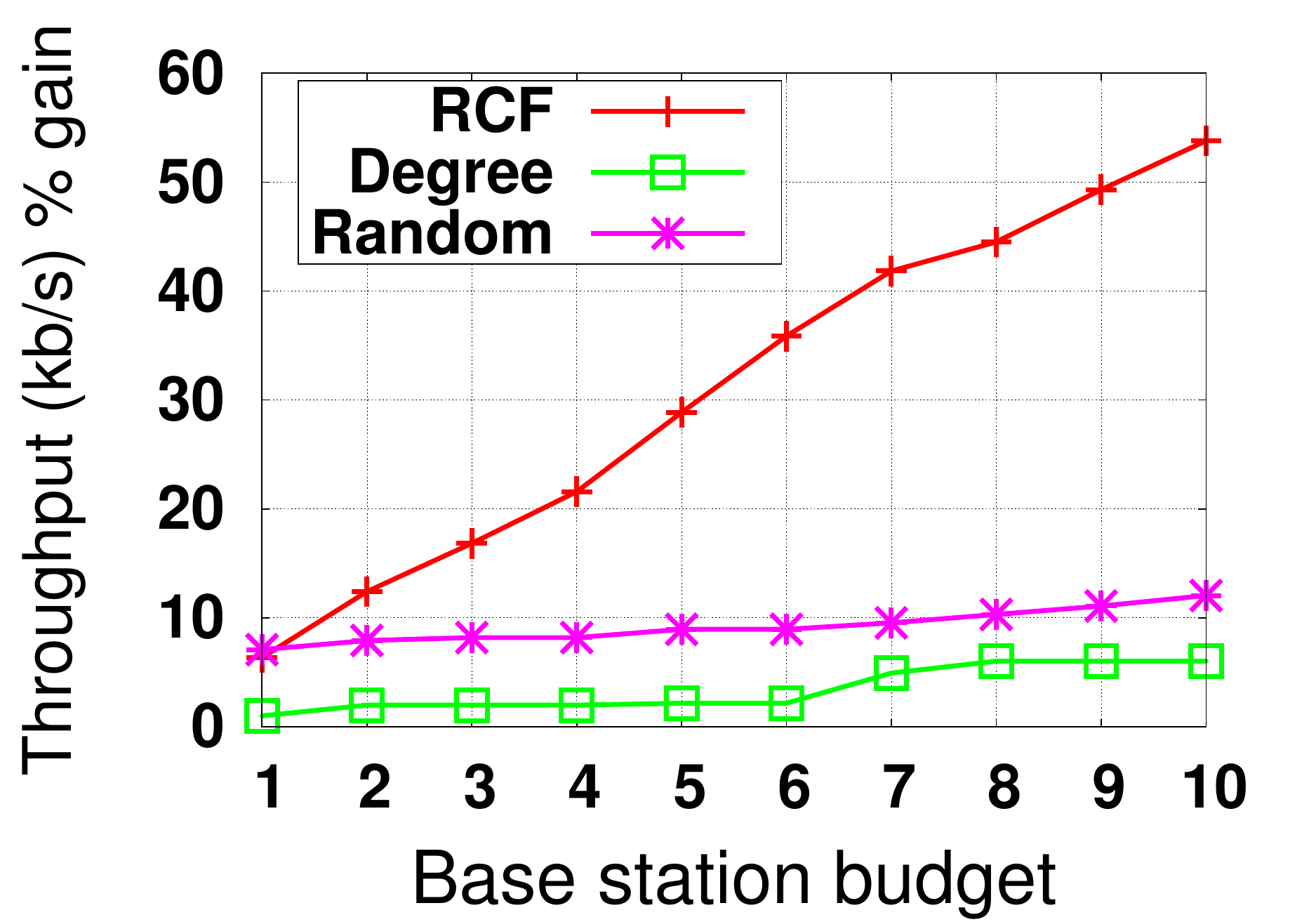}

\caption{BS incentivizes devices to block impact from rumors.\label{fig:bs_protect}}
\end{minipage}

\end{figure*}

Sample horizontal and vertical stripes are taken in the heat map and plotted in Figure \ref{fig:snap_shot}. The throughputs are shown for four different user awareness levels (1.0, 0.4, 0.2, 0.0).
As users become more prone to rumors, which is captured by smaller user awareness values, the network throughput decreases for fixed seed set size. On the other hand, when the number of seed users increases, the throughput decreases for a given user awareness.

\subsection{Reduce Vulnerability from the Cellular Network}\label{ssec: retention}
In Section \ref{ssec:awareness}, we demonstrate that the vulnerability can be reduced when the users can identify rumors about D2D and thus limit rumor propagation in the OSN. However, improve awareness for the general public may take time. In the meantime, the BS is responsible to ensure the usual network operation and continue to offer the best service. In this section, we discuss how the vulnerability can be reduced from the cellular network. To this end, the BS can identify the critical users to the D2D communication and incentivize them to stay using D2D. Using the same criticalness evaluation scheme, the BS can rank the devices based on their criticalness values. Once the top-ranked devices are identified, the BS offers those devices incentive to keep them in the cellular network. We now investigate how network throughput will change when the critical users are retained.

Figure \ref{fig:bs_protect} expresses the scenario where BS has a budget to incentivize the devices in the cellular network. In this experiment, we again use the stadium scenario only for conciseness and fix the number of critical nodes to $10$. We measure the performance of the incentive based approach in terms of throughput gain $G_l$ when the BS has budget $l$. Denote $T_{10}^{l}$ as the throughput calculated by RCF in the situation where the number of critical nodes is $10$ and the BS has budget of $l$. Define $G_l=\frac{G_l-T_{10}^{l}}{T_{10}^{l}}$, 
we can observe that the more budget the BS has, the larger throughput gain it achieves. Interestingly, this critical user retention scheme is much more effective against RCF compared with the other two approaches for critical nodes identification. The throughput gain is over $50\%$ with BS budget of $10$. This result suggests that incentives can be effective for reducing the vulnerability of the cellular network.

\section{Conclusion}\label{sec:conclusion}
In this paper, we investigated the vulnerability of D2D communication from interdependent OSN. To evaluate the vulnerability, we proposed the problem TMIN to find the critical nodes in the OSN. Then we developed a solution framework RCF to TMIN and discussed how it can be extended to consider general flow/diffusion networks. By experimental evaluation, we quantitatively demonstrated how vulnerable the cellular networks could be when rumors can propagate in the interdependent OSN. Also, we illustrated two possible approaches to reduce the vulnerability.

\bibliographystyle{unsrt}

\appendix
\section{Proofs}
\textbf{Proof for Theorem 1.}
\begin{proof}
We introduce the following two lemmas for the proof. 
\begin{lemma}\label{lemma:martingale}\cite{Tang15}
For any $\epsilon>0$, 
\begin{align*}\scriptsize
Pr[\sum_{i=1}^{T}x_i-T \mu_X\geq \epsilon T \mu_X]\leq e^{-\frac{\epsilon^2}{2+\frac{2}{3}\epsilon}T \mu_X}
\end{align*}
\end{lemma}
\begin{lemma}\label{lemma:lbsamples}\cite{Tang15}
If the number of RR sets is at least
\begin{align}\label{eqn:lbsamples}
Q^* =\frac{2|V^s|\phi^2}{\mathbb{I}^c(S^*)}
\end{align}
where
\begin{align*}
&\phi = \frac{(1-1/e)\sigma+\tau}{\epsilon}, \sigma =\sqrt{\ln(\frac{3}{\delta})}\\ 
&\tau = \sqrt{(1-\frac{1}{e})(\ln\binom{|V^s|}{k}+\ln\frac{3}{\delta})}
\end{align*}
The seed set $\bar{S}$ outputted by Alg. \ref{alg:tim} satisfies 
\begin{align*}
\mathbb{I}^c(\bar{S})\geq (1-\frac{1}{e}-\epsilon)\mathbb{I}^c(S^*)
\end{align*}
with probability at least $1-\frac{2\delta}{3}$.
\end{lemma}
Notice that Lemma \ref{lemma:lbsamples} is not in its original form, as we substitute $\delta$ with $\frac{2\delta}{3}$. 

Next, we show that when Alg. \ref{alg:tim} stops, the event that the total number of RR sets $|\mathcal{R}|$ generated is less than $Q^*$ can happen with probability at most $\frac{\delta}{3}$.

Here we consider Bernoulli random variables $X_{\bar{S}}^i=\min \{1,|\bar{S}\cap \tilde{R}|\}$ where $\tilde{R}$ is a random RR set. Due to the BSA sampling algorithm \cite{nguyen2016targeted}, the mean of $X_{\bar{S}}$ is $\mu_{X_{\bar{S}}} = \frac{\mathbb{I}^c(\bar{S})}{\Omega}$.
	\begin{align*}\scriptsize
		&\Pr[|\mathcal{R}|\leq Q^*] \leq \Pr[\sum_{i=1}^{|\mathcal{R}|}X_{\bar{S}}\leq \sum_{i=1}^{Q^*}X_{\bar{S}}]\\
        &=\Pr[\deg_{\mathcal{R}}(\bar{S})\leq \sum_{i=1}^{Q^*}X_{\bar{S}}]\leq Pr[\gamma\leq \sum_{i=1}^{Q^*}X_{\bar{S}}]\\
		&\leq Pr[\frac{(1+\log \frac{3}{\delta}\frac{1}{\phi^2})Q^*\times \mathbb{I}^c(\bar{S})}{\Omega}\times \frac{\Omega\mathbb{I}^c(S^*)}{|V^s|\mathbb{I}^c(\bar{S})}\leq \sum_{i=1}^{Q^*}X_{\bar{S}}]\\
        &\leq  Pr[\sum_{i=1}^{Q^*}X_{\bar{S}}-Q^*\mu_{X_{\bar{S}}}\geq \log \frac{3}{\delta}\frac{1}{\phi^2}Q^*\mu_{X_{\bar{S}}}\frac{\Omega\mathbb{I}^c(S^*)}{|V^s|\mathbb{I}^c(\bar{S})}] \\
        &\leq \exp(-\frac{(\log \frac{3}{\delta}\frac{1}{\phi^2})^2(\frac{\Omega\mathbb{I}^c(S^*)}{|V^s|\mathbb{I}^c(\bar{S})})^2}{2+\frac{2}{3}(\log \frac{3}{\delta}\frac{1}{\phi^2})\frac{\Omega\mathbb{I}^c(S^*)}{|V^s|\mathbb{I}^c(\bar{S})}}Q^* \mu_{X_{\bar{S}}}) \quad\text{(Lemma \ref{lemma:martingale})}\\
        &\leq \exp(-\frac{(\log \frac{3}{\delta}\frac{1}{\phi^2})^2\mathbb{I}^c(S^*)}{2\log \frac{3}{\delta}\frac{1}{\phi^2}|V^s|}Q^*)\\
        &\leq \exp(-\frac{\log \frac{3}{\delta}\frac{1}{\phi^2}\mathbb{I}^c(S^*)}{2|V^s|}Q^*)\leq \frac{\delta}{3}
\end{align*}
Combining Lemma \ref{lemma:lbsamples} and the result above, by union bound, we obtain the desired result for Theorem \ref{theorem:timratio}.
\end{proof}

\end{document}